\documentclass[10pt, conference, compsocconf]{IEEEtran}
\IEEEoverridecommandlockouts
\usepackage{amsmath}
\usepackage{amssymb}
\usepackage{lipsum}
\usepackage[font=footnotesize, labelfont=bf]{caption}
\usepackage{listings}
\usepackage{gensymb}
\usepackage{fancyhdr}
\usepackage{graphicx}
\newtheorem{thma}{Theorem}

\newtheorem{lemm}{Lemma}

\newcommand{\fracnot}[3]{\mathcal{F}_{#1}^{#2,#3}}

\begin{document}
%
\title{Block-space GPU Mapping for Embedded Sierpi\'nski Gasket Fractals}


\author{
\IEEEauthorblockN{Crist\'obal A. Navarro}
\IEEEauthorblockA{Institute of Informatics,\\
Universidad Austral de Chile,\\
Valdivia, Chile\\
Email: cnavarro@inf.uach.cl}
\and
\IEEEauthorblockN{Benjam\'in Bustos}
\IEEEauthorblockA{Department of Computer Science (DCC)\\
University of Chile, Santiago, Chile\\
Email: bbustos@dcc.uchile.cl}
\and
\IEEEauthorblockN{Raimundo Vega}
\IEEEauthorblockA{Institute of Informatics,\\
Universidad Austral de Chile,\\
Valdivia, Chile\\
Email: rvega@inf.uach.cl}
\and
\IEEEauthorblockN{Nancy Hitschfeld}
\IEEEauthorblockA{Department of Computer Science (DCC)\\
University of Chile, Santiago, Chile\\
Email: nancy@dcc.uchile.cl}
}

\maketitle

\begin{abstract}
This work studies the problem of GPU thread mapping for a Sierpi\'nski gasket fractal embedded in a discrete Euclidean space of $n \times n$. A block-space map
$\lambda: \mathbb{Z}_{\mathbb{E}}^{2} \mapsto \mathbb{Z}_{\mathbb{F}}^{2}$ is proposed, from Euclidean parallel space $\mathbb{E}$ to embedded fractal space
$\mathbb{F}$, that maps in $\mathcal{O}(\log_2 \log_2(n))$ time and uses no more than $\mathcal{O}(n^\mathbb{H})$ threads with $\mathbb{H} \approx
1.58...$ being the Hausdorff dimension, making it parallel space efficient.  When compared to a bounding-box map, $\lambda(\omega)$ offers a sub-exponential
improvement in parallel space and a monotonically increasing speedup once $n > n_0$. Experimental performance tests show that in practice $\lambda(\omega)$ 
can produce performance improvement at any block-size once $n > n_0 = 2^8$, reaching approximately $10\times$ of speedup for $n=2^{16}$ under optimal block
configurations.
\end{abstract}


%
\IEEEpeerreviewmaketitle

\section{Introduction}
Fractals can be described as self-similar structures \cite{mandelbrot2004} where a \textit{similar}\footnote{Depending on which fractal, the term \textit{similar}
can refer to \textit{exactly similar} or \textit{quasi similar}.} geometrical pattern is found at all scales.  Several natural phenomena produce fractal
patterns that obey a self-similar structure \cite{mandelbrot1982fractal}. Phenomena such as plant and tree growth \cite{Oppenheimer:1986:RTD:15886.15892,
Palmer1988}, terrain formation \cite{MILNE198867, 4767591}, molecular dynamics \cite{rothemund2004}, snowflake crystallization \cite{He2008739}, blood vessels
\cite{PhysRevLett.90.118101}, morphological features of living organisms \cite{WeibelL361}, among many others, display a fractal design where
self-similarity is a critical feature for modeling its geometrical structure. 

One well known fractal is the \textit{Sierpi\'nski Gasket}, or \textit{Sierpi\'nski Triangle}, described by Waclaw Sierpi\'nski in 1915.  Despite being over
than a century old, the Sierpi\'nski gasket remains a relevant subject as it is present in different fields such as the construction of antennas
\cite{855489, 664115}, cellular automata \cite{Ohi2001, RevModPhys.55.601}, molecular DNA self-organization \cite{shang2015,rothemund2004}, self-assembly theory
\cite{Doty:2012:TAS:2380656.2380675, LATHROP2009384} and phase transitions on fractal spin lattices \cite{0305-4470-17-2-028, PhysRevLett.45.855, Mota20086095},
among others, because of its special properties. For fractal simulations such as in Cellular Automata or Monte Carlo simulation on
spin models, the Sierpi\'nski gasket is usually discretized. Is this form, the fractal is defined as a self-similar structure where level
$r+1$ is composed of three repetitions of level $r$ at $1/2$ the scale,  as shown in Figure \ref{fig_sierpinski_discrete_construction}.  
\begin{figure}[ht!]
\centering
\includegraphics[scale=0.8]{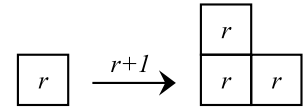}
\caption{Construction of the discrete Sierpi\'nski gasket.}
\label{fig_sierpinski_discrete_construction}
\end{figure}

Applications that involve graphical representations or data-parallel simulations with nearest neighbors interactions may benefit if the storage of the structure
preserves spatial locality in memory space, \textit{i.e}, the memory locations $(x\pm1,y\pm1)$ define a neighborhood in the actual fractal as well. 
One way to achieve this is to embed the fractal in a Euclidean space of $n \times n$ as shown in Figure \ref{fig_sierpinski_discrete_embedded}.  
\begin{figure}[ht!]
\centering
\includegraphics[scale=0.08]{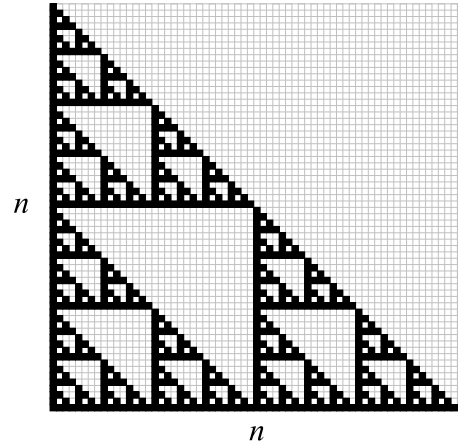}
\caption{A discrete Sierpi\'nski gasket embedded in a $n \times n$ Euclidean space.}
\label{fig_sierpinski_discrete_embedded}
\end{figure}

For applications that can operate in an embedded fractal domain, computations will usually consist of
accessing all the elements of the fractal and eventually perform arithmetic/logic operations that involve the data-element itself and possibly its nearest neighbors.
Eventually, when the fractal is large enough to the point of containing hundreds of thousands of elements, a sequential computation can take an excessive
amount of time for the practical requirements of the field. In these situations GPU computing becomes an attractive tool for accelerating the task \cite{navhitmat2014}. 

For every GPU computation there is a stage where threads are mapped from parallel to data space. A map, defined as $f: \mathbb{Z}^k \rightarrow \mathbb{Z}^m$,
transforms each $k$-dimensional point $x=(x_1, x_2, ..., x_k)$ in parallel space $P^k$ into a unique $m$-dimensional point $f(x) = (y_1, y_2, \cdots, y_m)$ in data
space $D^m$. GPU parallel spaces are defined as orthotopes $\Pi^k \in P^k$ in $k=1,2,3$ dimensions. A known way of mapping threads is to use the \textit{bounding-box}
(BB) approach, that builds an orthotope $\Pi^k$ sufficiently large to cover the whole data space and threads are mapped using the identity $f(\omega) = \omega$. 
Such map is highly convenient and efficient for the class of problems where data space is also defined by an orthotope; such as vectors, tables, matrices and
box-shaped volumes. But for an embedded fractal such as the Sierpi\'nski gasket, this approach is no longer efficient in terms of parallel space as many threads 
fall outside the domain, introducing a performance penalty to the execution time (see Figure \ref{fig_sierpinski_bounding_box_map}).
\begin{figure}[ht!]
\centering
\includegraphics[scale=0.16]{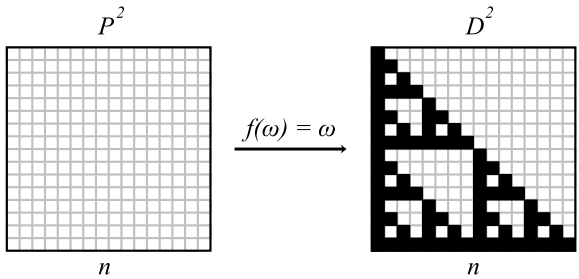}
\caption{In the bounding-box approach the threads that fall outside the fractal have to be discarded at run time.}
\label{fig_sierpinski_bounding_box_map}
\end{figure}

Two research questions arise from this GPU efficiency problem on the embedded Sierpi\'nski gasket; The first question: \textit{Is there any parallel-space efficient function,
namely $\lambda(\omega)$, that can use asymptotically the same number of threads as data elements in the fractal and map blocks properly onto the embedded Sierpi\'nski
gasket}?  (see Figure \ref{fig_sierpinski_efficient_map}).
\begin{figure}[ht!]
\centering
\includegraphics[scale=0.16]{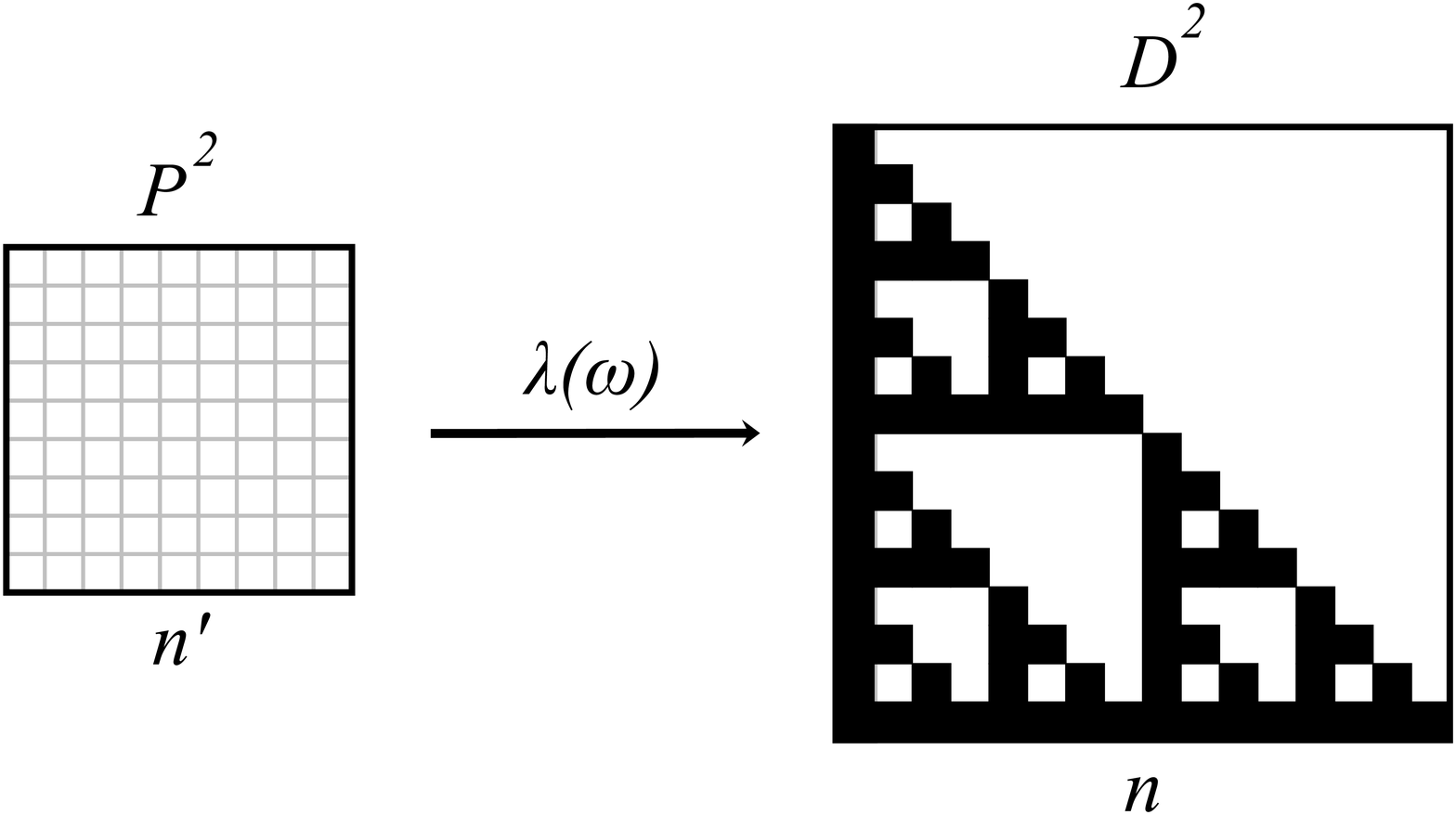}
    \caption{A $\lambda(\omega)$ map would use asymptotically the same number of threads as data elements.}
\label{fig_sierpinski_efficient_map}
\end{figure}

It is important to note that Question 1 asks for a block-space map and not a thread-space one.  The change from thread-space to block-space allows coalesced memory to be
preserved throughout the entire domain as thread organization is not compromised inside a block.  

The second question relates to performance: \textit{Will the parallel-space improvement translate into a significant GPU performance improvement}?

The present work focuses on these two questions and provides positive answers for both of them. A dedicated analysis is devoted to show that an alternating
unrolling strategy allows to define a parallel-space efficient $\lambda(\omega)$ that only requires $\mathcal{O}(n^\mathbb{H})$ threads, with $\mathbb{H} \approx
1.58...$ being the Hausdorff dimension of the Sierpi\'nski fractal. It addition, it is shown that by taking advantage of block-parallelism, $\lambda(\omega)$
becomes computable in $\mathcal{O}(\log_2 \log_2 (n))$ time which is fast enough to produce monotonically increasing speedup once $n > n_0$ with $n_0$
being a constant threshold value. 

The rest of the manuscript presents related work (Section \ref{sec_related-work}), a formal definition and analysis of $\lambda(\omega)$ (Section
\ref{sec_formulation}) and Section (\ref{sec_performance}) presents performance results. 
A discussion and comments on future work is found in Section \ref{sec_discussion}.

\section{Related Work}
\label{sec_related-work}
The following related works can be classified into two categories; (1) studies on efficient GPU mapping for triangular domains and (2) general studies on the
structure of the discrete Sierpi\'nski fractal.  

One of the first works that explored the possibilities of improving the GPU mapping and stage was the research of Jung \textit{et. al.} \cite{Jung2008} whom
proposed packed data structures for representing triangular and symmetric matrices with applications to LU and Cholesky decomposition
\cite{springerlink_gustavson}. The strategy is based on building a \textit{rectangular box strategy} for accessing and storing a triangular matrix (upper or
lower). Data structures become practically half the size with respect to classical methods based on the full matrix. The strategy was originally intended to
modify the data space (\textit{i.e.,} the matrix), however one can apply the concept analogously to the parallel space.

Ries \textit{et. al.} contributed with a parallel GPU method for the triangular matrix inversion \cite{Ries:2009:TMI:1654059.1654069}.  The authors identified
that the parallel space indeed can be improved by using a \textit{recursive partition} of the grid, based on a \textit{divide and conquer} strategy.  The
approach takes $O(\log_2(n))$ time by doing a balanced partition of the structure, from the orthogonal point to the diagonal. 

Q. Avril \textit{et. al.} proposed a GPU mapping function for collision detection based on the properties of the \textit{upper-triangular map} \cite{AvrilGA12}.
The map is a thread-space function $u(x) \rightarrow (a, b)$, where $x$ is the linear index of a thread $t_x$ and the pair $(a,b)$ is a unique two-dimensional
coordinate in the upper triangular matrix. Since the map works in thread space, the map is accurate only in the range $n \in [0, 3000]$ of linear problem size.

Navarro, Hitschfeld and Bustos have proposed a block-space map function for $2$-simplices and $3$-simplices \cite{DBLP:conf/hpcc/NavarroH14, CLEI-2016-navarro},
based on the solution of an $m$ order equation that is formulated from the linear enumeration of the discrete elements. The authors report performance
improvement for $2$-simplices, and for the $3$-simplex case, the mapping technique is extended to the discrete orthogonal tetrahedron, where the parallel space
usage can be $6\times$ more efficient. However the authors clarify that it is difficult to translate such space improvement into performance improvement, as the
map requires the computation of several square and cubic roots that introduce a significant amount of overhead to the process. From the point of view of
data-reorganization, a succinct blocked approach can be combined along with the block-space thread map, producing additional performance benefits with a
sacrifice of $o(n^3)$ extra memory. 

Exploring the benefits of efficient GPU mapping onto the embedded Sierpi\'nski gasket becomes an interesting topic of research since its geometry is no longer
Euclidean as in the related works, but instead it is embedded in an Euclidean one.  Finding a proper efficient $\lambda(\omega)$ would produce an asymptotic
improvement in parallel space and a potential performance improvement that could eventually be exploited.

\section{Analysis and Formulation of $\lambda(\omega)$}
\label{sec_formulation}

This Section first analyzes two important properties of the discrete embedded Sierpi\'nski gasket, which are helpful in the formulation of the map $\lambda(\omega)$.

\subsection{Analysis of the Space and Packing of $\fracnot{n}{k}{s}$}

The discrete Sierpi\'nski gasket belongs to a category of discrete fractals where its structure can be built by replicating $k$ instances of itself at a scale
$s$ and placed with an arbitrary spatial organization. The notation $\fracnot{n}{k}{s}$ is introduced to denote such fractal, where $n \in \mathbb{N}$ is the
linear size, $k \in \mathbb{N}$ the replication factor and $0 < s < 1\in \mathbb{R}$ the scaling factor in terms of reduction. Since fractals have a recursive
self-similar structure, their volume $\mathcal{V}(\fracnot{n}{k}{s})$ may be expressed as
\begin{equation}
    \mathcal{V}(\fracnot{n}{k}{s}) = \sum_{i=1}^{k} {\mathcal{V}(\fracnot{sn}{k}{s})} 
\end{equation}
with $\mathcal{V}(\fracnot{1}{k}{s}) = 1$ being the limit condition of the recursion. Given that the replication factor $k$ is fixed, and $n$ scales by
factors of $s$, the volume may be simplified into \begin{equation}
    \mathcal{V}(\fracnot{n}{k}{s}) = k^r
\end{equation}
where $r = \log_{1/s}(n)$ is defined as the scale level. 

The Sierpi\'nski gasket is a particular case where $k=3$, $s=1/2$ and $r = \log_2(n)$. Successive steps of the fractal produce the pattern early depicted in Figure
\ref{fig_sierpinski_discrete_embedded}. The following two lemmas are introduced to support the formulation of a map $\lambda(\omega)$.
\begin{lemm}
\label{lemma_hausdorff}
The space occupied by a discrete embedded Sierpi\'nski gasket is in correspondence with the Hausdorff dimension of the infinite Sierpi\'nski gasket.
\end{lemm}
\begin{proof}
    The space occupied by a Sierpi\'nski gasket of linear size $n$ is
    $\mathcal{V}(\fracnot{n}{3}{\frac{1}{2}}) = 3^r$.  Given that $r =
    \log_2(n)$ and $3^{log_2(x)} = 2^{\log_2(3)\log_2(x)}$, the space expression
    can be rearranged into\begin{equation}
        \mathcal{V}(\fracnot{n}{3}{\frac{1}{2}}) = n^{\mathcal{H}= \log_2(3)}
    \end{equation}
    where the exponent $\mathcal{H} = \log_2(3) \approx 1.5849...$ is the Hausdorff dimension of the original infinite Sierpi\'nski gasket.
\end{proof}

\begin{lemm}
\label{lemma_regular}
    A discrete Sierpi\'nski gasket of level $r$ packs into a $2$-orthotope $\Pi^2$ of dimensions $3^{\lceil \frac{r}{2} \rceil} \times 3^{\lfloor \frac{r}{2} \rfloor}$.
\end{lemm}
\begin{proof}
    Proof by induction on $r$:
    \begin{itemize}
        \item{Base case:} At scale $r=0$ the fractal has a space of $\mathcal{V}(\fracnot{1}{3}{\frac{1}{2}}) = 1$ element that packs into a regular
            $2$-orthotope of $1 \times 1 = 3^{\lceil \frac{0}{2} \rceil} \times 3^{\lfloor \frac{0}{2} \rfloor }$ satisfying $3^{\lceil \frac{r}{2} \rceil} \times 3^{\lfloor \frac{r}{2} \rfloor}$.
        \item{Induction step:} It is assumed that the orthotope for $r=k$ is quasi-regular. If $k$ is even, the packing for
            $k+1$ will triple the horizontal dimension of the $2$-orthotope. If $k$ is odd, the packing for $k+1$ will triple the vertical dimension of the
            $2$-orthotope. Since even and odd must alternate, the dimensions of the packed $2$-orthotope for $k+1$ can only be $3 \cdot 3^{\lceil{\frac{k}{2}}\rceil} 
            \times 3^{\lfloor{\frac{k}{2}}\rfloor}$ or $3\cdot 3^{\lceil{\frac{k}{2}}\rceil} \times 3\cdot 3^{\lfloor{\frac{k}{2}}\rfloor}$, which is regular or quasi-regular, respectively.
    \end{itemize}
\end{proof}

The packing process of Lemma (\ref{lemma_regular}) is illustrated in Figure \ref{fig_sierpinski_packing}, where the packing steps are in correspondence with the scale
levels shown in Figure \ref{fig_sierpinski_packing}.
\begin{figure}[ht!]
    \centering
    \includegraphics[scale=0.20]{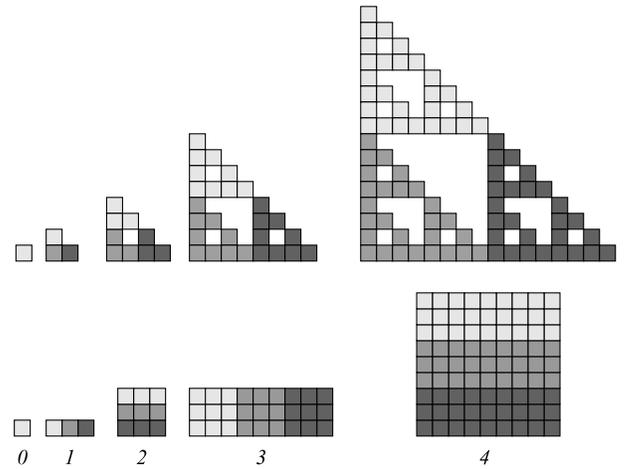}
    \caption{Each scale of the Sierpi\'nski fractal packs into a $2$-orthotope
    $\Pi^2$ of dimensions $3^{\lceil \frac{r}{2} \rceil} \times 3^{\lfloor \frac{r}{2} \rfloor}$.}
    \label{fig_sierpinski_packing}
\end{figure}

\subsection{Changing from Thread-space to Block-space}
An important aspect to consider is at which level the parallel space will be mapped. Two approaches are possible; (1) thread-space mapping and (2) block-space
mapping. For first one, $\lambda(\omega)$ defines $\omega$ as a unique thread location in parallel space.  For the second approach, $\lambda(\omega)$ defines
$\omega$ as a block coordinate in which several threads are contained.  The block-space approach has three important advantages over thread-space mapping.
First, in block-space, the fractal becomes a simplified version of the original, requiring less elements to be mapped. Second, since the fractal is a simplified
version of itself, it is possible to work on higher sizes of $n$ before the CUDA grid maximum dimensions are reached. Third, the block-space approach allows the
possibility for threads inside a block to preserve locality, which is essential for coalesced memory accesses on data.

The block-space map $\lambda(\omega)$ is now introduced, where $\omega = (\omega_x,\omega_y)$ denotes a two-dimensional coordinate of a block of constant size
$|B| = \rho_x \times \rho_y$ threads.  The change from thread-space to block-space means that blocks are mapped to a simplified version of the fractal of linear
size $n_b = n/b$, as shown in Figure \ref{fig_thread_vs_block}.
\begin{figure}[ht!]
\centering
\includegraphics[scale=0.05]{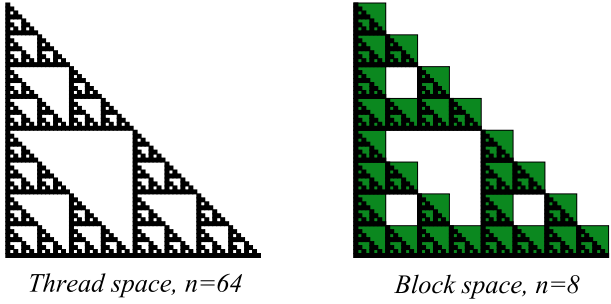}
    \caption{In thread-space mapping, threads are directly mapped one-to-one to the elements of the fractal of linear size $n=64$. In block-space mapping, $|B|
    = 8 \times 8$ and blocks of threads are mapped onto a simplified version of the fractal of linear size $n_b=64/8=8$.}
\label{fig_thread_vs_block}
\end{figure}

It is important to clarify that the extra green regions visible in Figure \ref{fig_thread_vs_block} do not necessarily mean unused threads. The stage of local thread
mapping is detailed later in this Section, where three approaches can be used.

The formulation of $\lambda(\omega)$ continues in block-space with $\rho_x = \rho_y$ to simplify the analysis, with $n_b$ the new simplified linear size of the
fractal with the origin $(0,0)$ located at the top-left corner for both the parallel and embedded fractal spaces, and with $y$ increasing downwards.

\subsection{Formulation of $\lambda(\omega)$}
The function $\lambda: \mathbb{Z}_{\mathbb{E}}^{2} \mapsto
\mathbb{Z}_{\mathbb{F}}^{2}$ is introduced as a mapping of block coordinates
$\omega$ in parallel-space onto block coordinates in the embedded space. The
intuition behind is an unrolling process applied in parallel to each $\omega \in \Pi^2$ through all
the scale levels. At each level, different $x,y$ offsets are accumulated to form
the final $(\lambda_x(\omega), \lambda_y(\omega))$ coordinate in the fractal.
\begin{thma}
\label{theorem_lambda}
    There exists a block-space parallel-space efficient $\lambda(\omega)$ that can map blocks in
    $\mathcal{O}(\log_2 \log_2 (n_b))$ time using $|B| = \mathcal{\theta}(\frac{\log_2(n_b)}{\log_2 \log_2(n_b)})$ threads per block. 
\end{thma}
\begin{proof}
    \textit{By construction}: let $r_b = 2^{n_b}$ be the block-space scale level of the fractal, $\Pi^2$ the $2$-orthotope of $3^{\lceil \frac{r_b}{2}
    \rceil} \times 3^{\lfloor \frac{r_b}{2} \rfloor}$ blocks that maps onto the discrete Sierpi\'nski gasket $\fracnot{n_b}{3}{\frac{1}{2}}$, with each block having
    $\rho_x \times \rho_y$ threads. By Lemma (\ref{lemma_hausdorff}), $\Pi^2$ is parallel-space efficient in block-space. 
    
    A helper index function $\beta_\mu(\omega)$ is defined as   
    \begin{equation}
        \beta_\mu(\omega) = \Big( \frac{\omega_x((\mu+1)\mod 2) + \omega_y(\mu \mod 2)}{3^{\lceil \frac{\mu}{2} \rceil-1}}\Big) \mod 3
    \end{equation}
    to produce an index in the range $\beta_\mu(\omega) = 0,1,2$ that identifies, within scale level $\mu \in [0..r_b]$, which of the three regions of the fractal does block
    $\omega$ belongs to. For even $\mu$, $\beta_\mu(\omega)$ acts on $\omega_x$. For odd $\mu$, it acts on $\omega_y$. Regions are sorted as $0$ for top, $1$ for
    middle and $2$ for right (see Figure \ref{fig_sierpinski_discrete_construction} for visual reference).

    Having the $\beta_\mu(\omega)$ index, the weight functions 
    \begin{align}
        \Delta_{\mu}^x = \Big\lfloor \frac{\beta_\mu}{2} \Big\rfloor,\ \ \ \Delta_\mu^y = \beta_\mu - \Big\lfloor \frac{\beta_\mu}{2} \Big\rfloor
    \end{align}
    compute the offset weights, $0$ or $1$, for each of the $x$ and $y$ directions at scale level $\mu$. The value of the offset corresponds to $2^{\mu-1}$, which
    is the linear size of each region at scale level $\mu$. For a given $\mu$, the combination of the weight functions with the offset produce partial coordinates of the
    form
    \begin{align} 
        \tau_x^\mu = \Delta_{\mu}^x2^{\mu-1},\\
        \tau_y^\mu = \Delta_{\mu}^y2^{\mu-1}
    \end{align}
    that contribute to the final mapped coordinate. The summation of all partial coordinates produce the map 
    \begin{align}
        \lambda(\omega) &= (\lambda_x(\omega), \lambda_y(\omega)), \\
        \lambda_x(\omega) &= \sum_{\mu=1}^{\log_2(n_b)} \tau_\mu^x\\
        \lambda_y(\omega) &= \sum_{\mu=1}^{\log_2(n_b)} \tau_\mu^y
    \end{align}
    which can be computed in $\mathcal{O}(\log_2\log_2(n))$ time (\textit{i.e.,} $n_b \in \theta(n)$) using a parallel reduction with the threads contained
    in the $\omega$ block. Finally, by Brent's Theorem \cite{Brent:1974:PEG:321812.321815}, $|B| = \mathcal{\theta}(\frac{\log_2(n)}{\log_2 \log_2(n)})$ threads are
    sufficient for a block of threads to reduce efficiently in parallel. 
\end{proof}

\begin{thma}
\label{theorem_speedup}
    $\lambda(\omega)$ requires asymptotically less work than the bounding-box approach. 
\end{thma}
\begin{proof}
   %
    The asymptotic work improvement factor of $\lambda(\omega)$ with respect to the bounding-box approach is the quotient of the costs of mapping all
    blocks using their corresponding $\Pi^2$ structures with the consideration $n_b \in \theta(n)$
    \begin{align}
        S_{\lambda(\omega)} &=      \frac{\mathcal{O}(1) \mathcal{V}(\Pi_{BB}^2)}{\mathcal{O}(\log_2\log_2(n))\mathcal{V}(\Pi^2_{\lambda(\omega)})}\\
    \end{align}
    where $\Pi^2_{BB}$ and $\Pi_{\lambda(\omega)}^2$ are the parallel-spaces for the bounding-box and $\lambda(\omega)$ approaches, respectively. The
    parallel-space of $\Pi_{BB}^2$ corresponds to the Euclidean box of $n_b \times n_b$ blocks, and the parallel-space of $\Pi_{\lambda(\omega)}^2$ is
    $\mathcal{O}(n^{\mathcal{H}})$ by Lemma (\ref{lemma_hausdorff}). Applying the limit $n\to\infty$
    \begin{align}
        \lim_{n\to\infty}{S_{\lambda(\omega)}} &= \lim_{n\to\infty} {\frac{\frac{\partial}{\partial n}(n^{2-\mathcal{H}})}{\frac{\partial }{\partial n}(\log_2\log_2(n))}}\\
                                               &= \lim_{n\to\infty} {\frac{(2-\mathcal{H})n^{1-\mathcal{H}}}{\frac{1}{n\log_2(n))}}}\\
                                               &= \infty
    \end{align}
\end{proof}
    The importance of Theorem (\ref{theorem_speedup}) is that it guarantees the existence of a $n > n_0$ where $\lambda(\omega)$ will start becoming 
    each time faster than the bounding box approach. A theoretical curve is presented in Figure
    \ref{fig_theoretical_improvements}.
    \begin{figure}[ht!]
    \centering
    \includegraphics[scale=0.72]{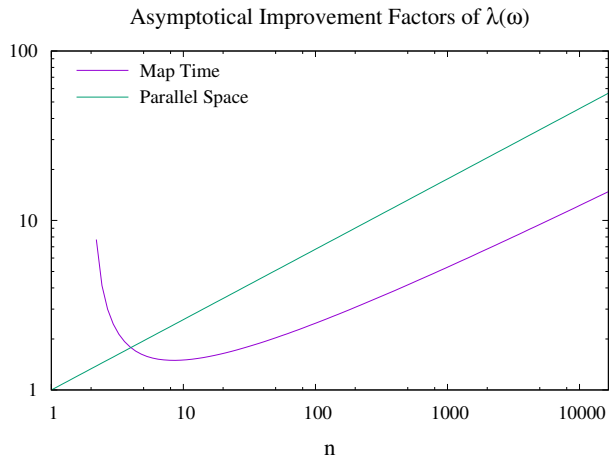}
        \caption{The theoretical improvement for parallel-space and mapping time.}
    \label{fig_theoretical_improvements}
    \end{figure}
    From the plot, one can note that space improvement is clear in the log-log
    scale. For the time improvement, the improvement decreases until $n_0 \sim
    7$, which is the point where it becomes a monotonically increasing function.
    In practice, constants could have an effect that would push $n_0$ further to
    the right.

\subsection{Intra-Block Mapping}
Once $\lambda(\omega)$ maps a block $\omega$, all the $\rho_x \times \rho_y$ threads contained have a shared reference coordinate that is available to use for
computing their individual location in the fractal. This phase of organizing the threads within a block is referred here as \textit{Intra-Block Mapping}, 
and this subsection describes three possible approaches to accomplish this.

\subsubsection{Further Unrolling}
In this approach threads inside their mapped block use the same
$\lambda(\omega)$ map but applied to each thread. By Theorem
(\ref{theorem_lambda}), the Intra-block map is parallel-space efficient 
and the mapping time becomes $\mathcal{O}(\log_2 \log_2(|B|)) \in \mathcal{O}(1)$ as 
$\rho_x, \rho_y$ are constant and do not grow with $n$.

\subsubsection{Shared Lookup Table}
The second approach is to use a shared lookup table of $\rho_x \times \rho_y =
\mathcal{O}(1)$ offset coordinates, available to any thread in any
block. Mapping each thread would cost $\mathcal{O}(1)$ memory accesses and the 
extra memory introduced by the shared table is $\mathcal{O}(\rho_x \times \rho_y) \in \mathcal{O}(1)$.

\subsubsection{Bounding Sub-boxes}
The third approach consists of using small bounding-boxes in each block. 
This approach introduces a constant number of extra threads in each block,
but allows each thread to be mapped just with $f(x) = x$ which costs
$\mathcal{O}(1)$. In order to know if it is in the fractal or not, 
each thread evaluates if $t_x \& (n-1-t_y) == 0$ is true or not, respectively, 
with $\&$ being the bitwise AND operator.

Regardless of which Intra-block mapping approach is chosen, the final mapping time will not surpass the $\mathcal{O}(\log_2\log_2(n))$ time, as the blocks have
a constant size of threads. Nevertheless, it is worth considering that \textit{Further Unrolling} introduces a constant cost in mapping time.  The
\textit{Shared Lookup Table} approach introduces a constant cost in memory and finally the \textit{Bounding Sub-boxes} introduce a constant in the number of
extra threads. Choosing one or another can depend on the specific application, \textit{i.e.}, to avoid competing with the application in the use of
memory bandwidth or arithmetic operations.

\section{Implementation and Performance Results}
\label{sec_performance}
A CUDA implementation of $\lambda(\omega)$ is put under test to obtain the actual speedup for different values of $n$. The implementation uses the
\textit{Bounding Sub-boxes} approach to arrange threads inside each block, as it is the simplest in terms of implementation time. The parallel reduction
per-block is performed using the shuffle instruction of CUDA, which allows efficient register-level communication among threads within the same warp.

The performance test consists of measuring the average time taken to write a constant value on all the elements of a Sierpi\'nski gasket of scale level $r$,
which is embedded in a $M_{n\times n}$ matrix filled with zeros. The configuration space is explored in the ranges $r=0..16$ and $\rho = 1,2,4,8,16,32$ for the
scale level and block-size respectively, in order to find the optimal setting that provides the best performance for both the bounding-box and $\lambda(\omega)$ approaches.
The average performance measures are taken by averaging 100 sub-averages, each one being an average time of 10 consecutive synchronized kernel calls. The
standard error for each mean was below $1\%$. The hardware for performance test is listed in Table \ref{table_hardware}.
\begin{table}[ht!]
\normalsize
\caption{Hardware used for performance tests.}
\begin{center}
\begin{tabular}{|c|r|}
\hline
Device	&	Model\\
\hline
GPU	&	Titan-X Pascal, GP102, 3584 cores, 12GB\\
CPU	&	Intel i7-6950X 10-core Broadwell\\
RAM	&	128GB DDR4 2400MHz\\
\hline
\end{tabular}
\end{center}
\label{table_hardware}
\end{table}

Figure \ref{fig_performance} presents the speedup of $\lambda(\omega)$ over the bounding-box approach, as well as the running times for the two mapping techniques.
\begin{figure*}[ht!]
\centering
\includegraphics[scale=0.72]{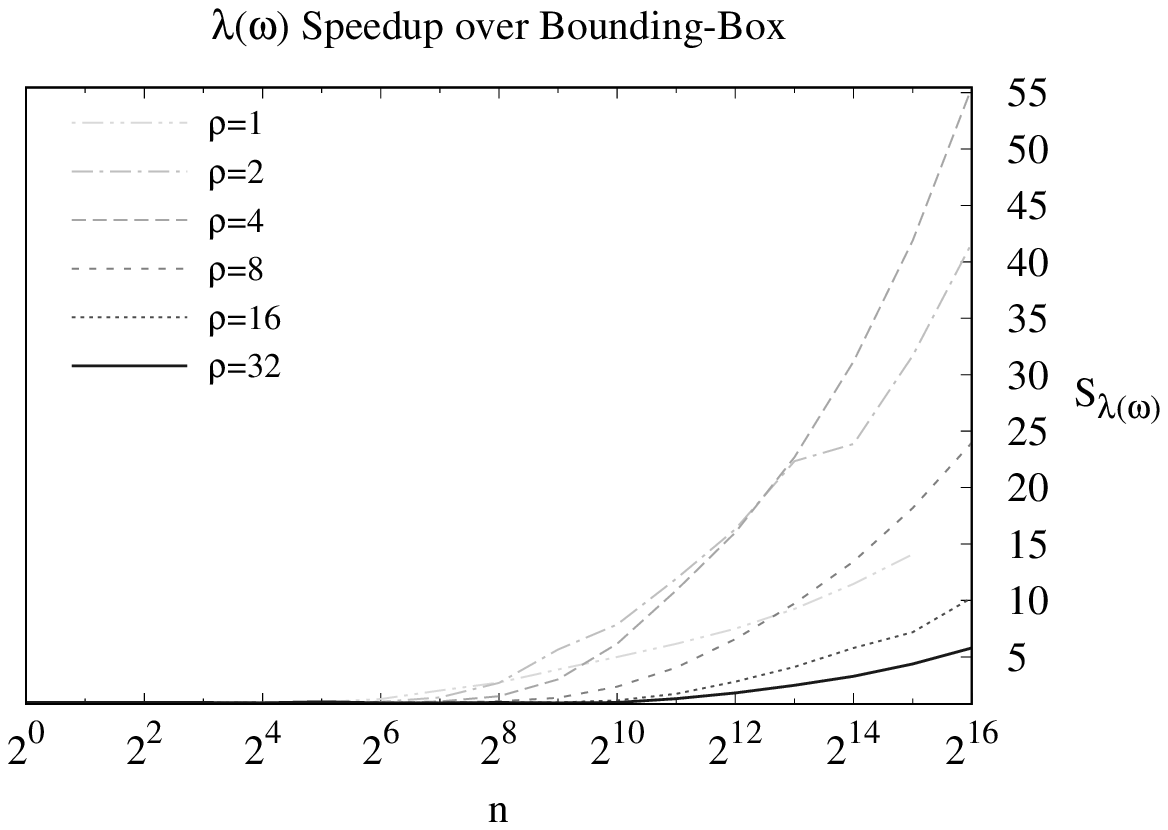}
\includegraphics[scale=0.72]{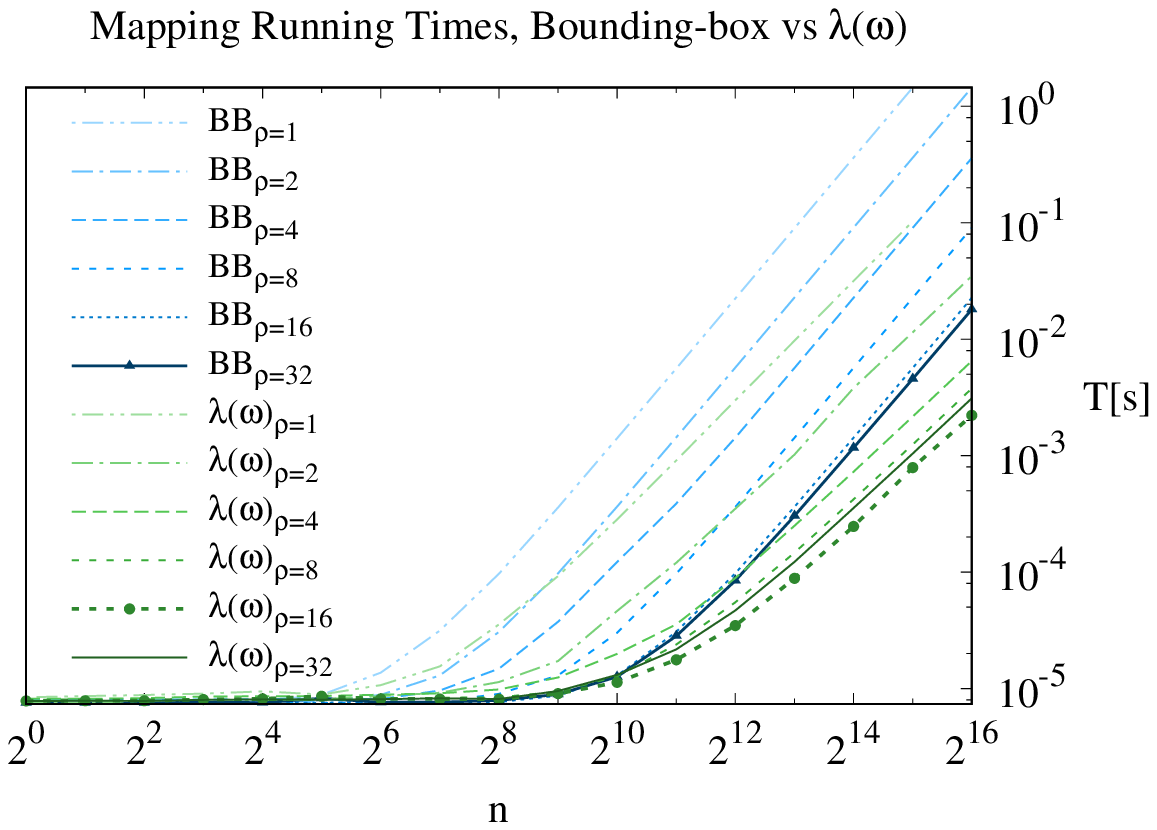}
    \caption{On the left, the speedup of $\lambda(\omega)$ with respect to the bounding-box approach at different block-size configurations. On the right, their
    absolute running times at different block-size configurations.}
\label{fig_performance}
\end{figure*}
For values of $n < 2^8$, one can note that only some curves offer speedup. Once $n > 2^8$, the speedup begins to increase for all block-size configurations,
reaching the higher values at $n=2^{16}=65536$, which was the highest problem size that fit in the GPU memory. An important aspect to note from the speedup curve
is that for the largest possible block size, $|B|=\rho \times \rho=32 \times 32$, the $\lambda(\omega)$ map runs the test approximately
$6\times$ faster than the bounding-box approach. Furthermore, as blocks become smaller in $\rho$, the improvement increases dramatically, reaching up to $55\times$ of
speedup. 

The plot of the running times provides further insights on what configuration is the best suited for each mapping technique. By looking at the running times of
the small block configurations, one can note that regardless of their high speedup, their running times are the slowest ones. For the bounding-box approach (BB)
the best performance is obtained when the block-size is maximum. For $\lambda(\omega)$ the best performance is found when using a block of $|B| = 16\times 16$
threads. If the curves of the best configuration for each implementation are considered, \textit{i.e.}, the ones with marked points on Figure
\ref{fig_performance}, right, then the speedup provided by $\lambda(\omega)$ increasing further more reaching almost an order of magnitude. The other running
time curves are still useful to visualize that as blocks become smaller, the value of $n_0$ where $\lambda(\omega)$ is more convenient moves closer to the
origin, and vice versa, however the GPU with its current organization and architecture is not fully utilized by such small block configurations, thus leading to
an inferior performance. Therefore, in practice large blocks should be utilized and by Theorem (\ref{theorem_speedup}), beyond $n = 2^{16}$ the speedup would
keep increasing in favor of $\lambda(\omega)$.

\section{Discussion}
\label{sec_discussion}
The results obtained in this work have shown that a parallel-space efficient mapping function can lead to significant performance gains in applications that
require processing an embedded Sierpi\'nski gasket.  The analysis and formulation of $\lambda(\omega)$ has provided three important results; (1) There exists a
correspondence between a quasi-regular $2$-orthotope of discrete elements and the elements of the embedded Sierpi\'nski gasket fractal. (2) Such correspondence
from parallel to data space can be computed in $\mathcal{O}(\log_2 \log_2(n))$ time using a block-space map that is based on efficient parallel reductions. (3)
The total work of mapping the $2$-orthotope with $\lambda(\omega)$ is asymptotically smaller than the work generated by the bounding box approach, leading to a
monotonically increasing speedup that is guaranteed to occur once a $n_0$ value is reached.

The experimental performance results confirm the theory as once the fractal reaches the linear size $n_0 = 2^8$, the
speedup begins to increase monotonically for all block-sizes. Using the maximum block size of $32 \times 32$ threads, $\lambda(\omega)$ 
reached up to $6\times$ of speedup. The maximum speedup obtained was approximately $55\times$ with blocks of $1$ thread, however
such block configuration is not practical for the GPU architecture. Still, having measured the running times with  
different block configurations helped in understanding that reducing the block size only brings the $n_0$ value closer to the
origin and increasing the block-size pushes it forward. Eventually, all block configurations can reach up to $55\times$ of speedup 
and beyond if the fractal is large enough. 

The GPU map found for the embedded Sierpi\'nski gasket may be adapted to work for other types of embedded fractals that follow a
similar building scheme applying modifications to the helper, weight and offset functions. In order to obtain speedup, it is
crucial to check if the overall work will be asymptotically smaller than the bounding-box approach.

Two important questions may be formulated from the results obtained in this work. The first one is: \textit{Can there exist a
general $\lambda(\omega)$ that maps a family of embedded fractals who share the same building principle?}, and the second
question: \textit{can there exist a $\lambda(\omega)$ that maps in $\mathcal{O}(1)$ time using no more than $\mathcal{O}(1)$ extra memory?}.
Future research in these directions can provide important insights on the limits of efficient GPU computing for embedded fractal domains.

\section*{Acknowledgment}
This project was supported by the research project FONDECYT N$^o$ 3160182 from
CONICYT, as well as by the Nvidia CUDA Research Center at the Department of 
Computer Science (DCC) from University of Chile.



%
\bibliographystyle{plain}
\bibliography{fractal-map}

\end{document}